\documentclass[runningheads]{article}

\usepackage{amsmath}
\usepackage{enumerate}
\usepackage[english]{babel}
\usepackage{amssymb}
\usepackage{amscd}
\usepackage{authblk}

\usepackage{url}
\usepackage{stmaryrd}
\usepackage{tikz,pgf}
\usepackage{subfig}
\usepackage{cancel}
\usepackage{comment}
\usepackage[hidelinks, draft = false]{hyperref}
\usepackage[textsize=footnotesize]{todonotes}
\usepackage{complexity}
\usepackage{amsthm}

\usepackage{cleveref}
\usetikzlibrary{positioning}
\usetikzlibrary{shapes.geometric}
\usepackage{soul}
\usepackage{enumitem}
\usepackage{xspace}
\usetikzlibrary{decorations.pathmorphing}
\usetikzlibrary{arrows}
\usetikzlibrary{decorations.pathreplacing,calligraphy}
\definecolor{myBlue}{HTML}{2D2F92}

\definecolor{FusionRed}{HTML}{fc5c65}
\definecolor{BenukionBronze}{HTML}{e55039}
\definecolor{HighBlue}{HTML}{3c6382}
\definecolor{RoyalBlue}{HTML}{3867d6}
\definecolor{AlgalFuel}{HTML}{20bf6b}
\definecolor{GloomyPurple}{HTML}{8854d0}
\definecolor{Innuendo}{HTML}{a5b1c2}

\newtheorem{theorem}{Theorem}
\newtheorem{lemma}[theorem]{Lemma}

\newtheorem{corollary}[theorem]{Corollary}

\newcommand{\qedclaim}{\hfill $\diamond$ \smallskip}

\newcommand{\etal}{{\em et al.}\xspace}
\tikzset{snake it/.style={decorate, decoration={snake, amplitude=.5mm}}}

\usepackage[T1]{fontenc}
\usepackage{graphicx}

\begin{document}
\title{An Algorithm for Monitoring Edge-geodetic Sets in Chordal Graphs}
%
%
\author[1]{Clara Marcille}
\author[2]{Nacim Oijid\thanks{The second author was partly supported by the Kempe Foundation Grant No. JCSMK24-515 (Sweden).} }

\affil[1]{ Univ Lyon, EnsL, UCBL, CNRS, LIP, F-69342, LYON Cedex 07, France}

\affil[2]{Department of Mathematics and Mathematical Statistics, Umeå University, Sweden}
\maketitle              
\begin{abstract}
A \textit{monitoring edge-geodetic set} (or meg-set for short) of a graph is a set of vertices $M$ such that if any edge is removed, then the distance between some two vertices of $M$ increases. This notion was introduced by Foucaud \textit{et al.} in 2023 as a way to monitor networks for communication failures. As computing a minimum meg-set is hard in general, recent works aimed to find polynomial-time algorithms to compute minimum meg-sets when the input belongs to a restricted class of graphs. Most of these results are based on the property of some classes of graphs to admit a unique minimal meg-set, which is then easy to compute. In this work, we prove that chordal graphs also admit a unique minimal meg-set, answering a standing open question of Foucaud \textit{et al.}
\end{abstract}

\section{Introduction and notations}
In the area of network monitoring problems, graphs are used to emulate networks and their possible behaviors. For instance, there is a well-established way to represent a network by a graph where each terminal is represented as a vertex, such that two vertices are adjacent if the corresponding terminals can communicate directly. In this representation, deleting an edge can be interpreted as the failure of a communication, see for instance~\cite{bampas2015network,beerliova2005network,bejerano2003robust,foucaud2022monitoring}. Monitoring edge-geodetic sets (or meg-sets for short) were introduced by Foucaud \textit{et al.} in a recent paper~\cite{foucaud2023monitoring}, as a way to detect such failures.
For a graph $G=(V, E)$, an edge $e\in E$ and two vertices $a, b\in V$, we say that $\{a, b\}$ \textit{monitors} the edge $e$ if it lies on all the shortest paths between $a$ and $b$. Note that deleting $e$ would increase the distance between $a$ and $b$, which is detected. By extension, we say $M\subset V$ monitors $e$ if there exist two vertices of $M$ monitoring $e$. Finally, we say that $M$ is a \textit{meg-set} if it monitors all edges of $E$. As stated in~\cite{foucaud:hal-04494089}, this can be understood as choosing a set of ``probes'' in the network which measure the distance to one another, and doing so, detecting that a failure in the network happened if one of the distance increases. As such measurements can be costly, we want to minimize the number of probes, that is, find the size of a minimum meg-set (w.r.t. size). This problem, known as \textsc{MEG-set}, is hard in general~\cite{haslegrave2023monitoring}, but can be solved in polynomial time on some classes of graphs using a characterization from~\cite{foucaud:hal-04494089}. 

A vertex $u\in V(G)$ is \textit{mandatory} if it belongs to all meg-sets of $G$, and $Mand(G)$ denotes the set of mandatory vertices. In the same paper, the authors prove that computing the set of mandatory vertices can be done in polynomial time. This computation also provides a polynomial-time algorithm to compute the meg-sets of graphs which are \textit{meg-minimal}, that is, with the property of admitting a unique minimal meg-set (w.r.t. inclusion), which is thus the set of mandatory vertices. These graph classes include namely interval graphs, cographs, block graphs, and well-partitioned chordal graphs~\cite{foucaud:hal-04494089}. In a more recent work~\cite{foucaud2025characterizingoptimalmonitoringedgegeodetic}, the authors prove that strongly chordal graphs are also meg-minimal. We refer the reader to the corresponding papers for a formal definition of these classes, and give an overview of the inclusion relations between these classes in Figure~\ref{fig:MEGclasses}. As all the above-mentioned graph classes (except cographs) are subclasses of chordal graphs, the next step is to inquire whether or not chordal graphs are meg-minimal. In particular, the best known algorithm on chordal graphs from~\cite{foucaud2024algorithms} is an \textsc{FPT} algorithm parameterized by treewidth plus solution size.

\begin{figure}[h!]
    \centering
    \begin{tikzpicture}[inner sep = .5mm, scale =.9]
    \tikzstyle{every node}=[ellipse, minimum width=60pt, align=center]
    \node[ellipse, draw, rounded corners, fill=green!70, minimum height=1cm, align=center] (trees) at (4,-0.5) {\scriptsize Trees~\cite{foucaud2023monitoring}};

    \node[ellipse, draw, rounded corners, fill=green!70,  minimum height=1cm, align=center] (complete) at (1,-1.5) {\scriptsize Complete~\cite{foucaud2023monitoring}};

    \node[ellipse, draw, rounded corners, fill=green!70,  minimum height=1cm, align=center] (split) at (-2,2.5) {\scriptsize Split~\cite{foucaud:hal-04494089}};

    \node[ellipse, draw, rounded corners, fill=green!70,  minimum height=1cm, align=center] (propint) at (-5,-.5) {\scriptsize Proper Interval\\ \cite{foucaud:hal-04494089}};

    \node[ellipse, draw, rounded corners, fill=green!70,  minimum height=1cm, align=center] (int) at (-5,1) {\scriptsize Interval~\cite{foucaud2024algorithms}};

    \node[ellipse, draw, rounded corners, fill=green!70,  minimum height=1cm, align=center] (strchordal) at (-5,2.5) {\scriptsize Strongly Chordal\\ \cite{foucaud2025characterizingoptimalmonitoringedgegeodetic}};

    \node[ellipse, draw, rounded corners, fill=green!70,  minimum height=1cm, align=center] (WPchordal) at (4,2.5) {\scriptsize W-P chordal \\ \cite{foucaud:hal-04494089}};

    \node[ellipse, draw, rounded corners, fill=green!70,  minimum height=1cm, align=center] (cograph) at (1, 4.5) {\scriptsize Cographs\\ \cite{foucaud:hal-04494089}};

    \node[ellipse, draw, rounded corners, fill=blue!30,  minimum height=1cm, align=center] (chordal) at (-2,4) {\scriptsize Chordal};

    \node[ellipse, draw, rounded corners, fill=orange!30,  minimum height=1cm, align=center] (perfect) at (-.5,6) {\scriptsize Perfect};

    \draw[->] (int.south) to (propint.north)[]{};
    \draw[->] (propint.east) to (complete.west)[]{};
    \draw[->] (strchordal.south) to (int.north)[]{};
    \draw[->] (chordal.south west) to (strchordal.north east)[]{};
    \draw[->] (chordal.south east) to (WPchordal.north west)[]{};
    \draw[->] (WPchordal.south) to (trees.north)[]{};
    \draw[->] (WPchordal.south west) to (complete.north east)[]{};
    \draw[->] (chordal.south) to (split.north)[]{};
    \draw[->] (perfect.south west) to (chordal.north)[]{};
    \draw[->] (perfect.south east) to (cograph.north)[]{};
    \draw[->] (cograph.south) to (complete.north)[]{};
    \draw[->] (split.south) to (complete.north west)[]{};
    
    \end{tikzpicture}
    \caption{A representation of the inclusion relations between the classes of graphs mentioned in the introduction. Our contribution is shown in blue, and the graph classes that are known not to be meg-minimal are in orange (perfect graphs). All the other depicted graph classes were previously known to be meg-minimal.}
    \label{fig:MEGclasses}
\end{figure}

In this work, we answer the question of the complexity of finding a minimum meg-set when the input graph is chordal by proving the following theorem:
\begin{theorem}\label{thm:mainthm}
    Let $G$ be a chordal graph. Then $Mand(G)$ is a meg-set of $G$.
\end{theorem}
In other words, we prove that \textsc{MEG-set} can be solved in polynomial time on chordal graphs.
This result can be seen as a more fine-grained proof of Theorem 6.1 of~\cite{foucaud2025characterizingoptimalmonitoringedgegeodetic}, generalized to chordal graphs, thus requiring more work to handle the new possibilities. In particular, they use a strong elimination ordering for their induction, while we generalize to a perfect elimination ordering. In addition to that, we want to acknowledge that some technical lemmas are similar between the two proofs. For instance, Claim 6.3 of~\cite{foucaud2025characterizingoptimalmonitoringedgegeodetic} is very close in spirit to Lemma~\ref{lemma:u-mand-in-G-v-equiv-u-mand-in-G}.

We remind that if $G$ is a graph, a vertex $v$ of $G$ is said to be {\em simplicial} if its neighborhood is a clique. Also a graph $G$ is said to be chordal if it has a perfect elimination ordering, i.e., an ordering of its vertices $v_1, \dots, v_n$ such that for any $1 \le i \le n-1$, $N(v_i) \cap\{v_{i+1}, v_{i+2}, \dots, v_n\}$ is a clique, where $N(v_i)$ is the neighborhood of the vertex $v_i$.

Throughout this paper, we denote a path on the vertices $a_1, a_2, \dots, a_k$ by the concatenation $a_1a_2\dots a_k$. We also denote a cycle the same way, though repeating the first vertex at the end. For instance, the cycle of size three on $a, b, c$ is denoted $abca$. Finally, to be coherent with previous works on meg-sets, we call an \textit{induced $2$-path} $abc$ in a graph $G$ a set of three vertices $a, b, c$ of $G$ such that $abc$ is an induced $P_3$.
\section{Some general lemmas}

\begin{theorem}[Foucaud \etal~{\cite[Theorem 4.1]{foucaud:hal-04494089}}]\label{theorem:allmeg}
Let $G$ be a graph. A vertex $v\in V(G)$ is mandatory if and only if there exists $w \in N(v)$ such that for any vertex $x\in N(v)$, the induced $2$-path $wvx$ is part of a $4$-cycle.
\end{theorem}
If it exists, such a vertex $w$ is called a \textit{support} of $v$.

\begin{lemma}\label{lemma: simplicial vertices are mandatory}
Let $G$ be a graph, and $v$ be a simplicial vertex of $G$. Then $v \in Mand(G)$.
\end{lemma}

Lemma~\ref{lemma: simplicial vertices are mandatory} is a direct consequence of Theorem~\ref{theorem:allmeg}, but was already mentioned earlier in~\cite{foucaud2023monitoring}.

\begin{lemma}\label{lemma:induced-2-path-in-a-cycle-has-a-chord}
    If $G = (V,E)$ is a chordal graph, and $a,b,c \in V$ are three vertices such that $abc$ is an induced $2$-path of $G$, which is in a cycle $C$. Then there is a \textit{chord} in $C$ incident to $b$,  that is, an edge between $b$ and some vertex of $C$ different from $a$ and $c$. 
\end{lemma}

\begin{proof}
  By definition, $C$ is a $C_k$ for $k \ge 4$ of $G$. Thus, it is necessarily not induced (as $G$ is chordal and hence admits no induced $C_k$). Hence, there is a chord in $C$. Either it is adjacent to $b$ or it produced a smaller cycle $C'$ containing $\{a,b,c\}$. Up to applying the same argument to $C'$, as $abc$ is an induced $P_3$, $a$ and $c$ cannot be adjacent. Hence, when we cannot find such a $C'$, we obtain a chord adjacent to $b$.
\end{proof}

\begin{lemma}[Foucaud \etal~{\cite[Lemma 2.3]{foucaud2023monitoring}}]\label{lemma:cut-vertices}
    Let $G$ be a graph with a cut-vertex G and $C_1, C_2, \dots , C_k$ be the $k$ components obtained
when removing $v$ from G. If $S_1, S_2, \dots , S_k$ are meg-sets of the induced subgraphs $G[C_1\cup\{v\}], G[C_2\cup\{v\}], \dots, G[C_k\cup\{v\}]$, then $S=(S_1\cup \dots \cup S_k)-v$ is an meg-set of $G$.
\end{lemma}

\section{Proof of the main theorem}

We consider $G$ a chordal graph which is a minimal counter-example (w.r.t. the number of vertices) to Theorem~\ref{thm:mainthm}, and we prove that $G$ cannot exist. To achieve this, let $v$ be a simplicial vertex of $G$. By minimality of $G$, since  $G -\{v\}$ is chordal, it is not a counter example. Hence, $M' = Mand(G -\{v\})$ is a meg-set of $G -\{v\}$. Let $W \subset V(G -\{v\})$ such that for all $w \in W$, $w$ is mandatory in $G -\{v\}$ but not in $G$. Let $M = \{v\} \cup M'\setminus W$. We prove that $M$ only contains mandatory vertices and that it is a meg-set of $G$, contradicting the fact that it is a counterexample.

\subsection{Overview of the proof}
The proof of Theorem~\ref{thm:mainthm} considers a minimum counter-example $G$ (w.r.t. the number of vertices) to its statement, and apply the well-known elimination rule of chordal graphs, that is, we consider a simplicial vertex $v$ of $G$ and remove it. This yields a chordal graph $G'$ for which Theorem~\ref{thm:mainthm} holds. In Lemma~\ref{lemma:e-monitored-by-ab-nonadjacent-to-v}, we prove that the mandatory vertices of $G -\{v\}$ keep monitoring in $G$ the edges they monitor in $G -\{v\}$, that is, if an edge of $G -\{v\}$ is not monitored in $G$, then it is monitored by at least one vertex of $W=Mand(G -\{v\})\cap (G-Mand(G))$. We then prove in Lemma~\ref{lemma:u-mand-in-G-v-equiv-u-mand-in-G} that vertices of $W$ all belong to the neighborhood of $v$.

\subsection{Preliminaries}

Next lemma states that adding a simplicial vertex cannot create a shortcut.

\begin{lemma}\label{lemma:e-monitored-by-ab-nonadjacent-to-v}
    If $a,b \in V(G) -\{v\}$ are two vertices such that $\{a,b\}$ monitors an edge $e$ (in $G -\{v\}$), then $\{a, b\}$ monitors $e$ in $G$.
\end{lemma}

\begin{proof}
  Let $e \in E(G)$ be an edge monitored by $\{a,b\}$ where $a, b\in V(G) -\{v\}$, then it means that $e$ is on every shortest path between $a$ and $b$. Let $P$ be one of these shortest paths. Since $v$ is simplicial, all its neighbors in $P$ are adjacent. Thus, $v$ has at most two neighbors in $P$ and if so, they are consecutive on the path, otherwise, the path would have a chord. Thus, $P$ is still a shortest path in $G$. Moreover, a shortest path between $a$ and $b$ in $G$ cannot go through $v$, otherwise, it will contain $v$ and two of its neighbors. Since $v \notin \{a,b\}$, it would contain an induced triangle, which cannot happen for a shortest path. Finally, $e$ is in all the shortest path between $a$ and $b$ in $G$, and therefore is monitored by $\{a,b\}$.
\end{proof}

\begin{lemma} \label{lemma:u-mand-in-G-v-equiv-u-mand-in-G}
    If $u\notin N(v)$, then $u\in Mand(G -\{v\})$ if and only if $u\in Mand(G)$.
\end{lemma}
\begin{proof}
    Suppose first $u \in Mand(G)$. There exists a support vertex $w$ of $u$. Hence, for any vertex $x \in V(G)$, such that $wux$ is an induced $2$-path, there exists $v'$ such that $wuxv'w$ is a $4$-cycle. Moreover, we must have $v' \ne v$, otherwise, since $v$ is simplicial, $wx$ would be an edge, and thus, $wux$ would not be an induced $2$-path. Since $v\notin N(u)$, we have $w \ne v$. Thus, for any vertex $x \in V(G) - \{v\}$, if $wux$ is an induced $2$-path, the same vertex $v'$ ensures that $wuxv'w$ is a $4$-cycle, and thus $w$ is a support for $u$, and $u \in Mand(G -\{v\})$.

    Reciprocally, suppose that $u \in Mand(G -\{v\})$, and let $w$ be a support for $u$ in $G -\{v\}$. Since $v \notin N(u)$, any induced $2$-path $wux$ of $G$ is an induced $2$-path of $G -\{v\}$. Thus, if $v'$ is such that $wuxv'w$ is a $4$ cycle in $G -\{v\}$, $wuxv'w$ is also a $4$-cycle of $G$. Thus, $u$ is mandatory in $G$.
\end{proof}

\begin{corollary}\label{corr:W}
    $W \subset N(v)$.
\end{corollary}

Recall that $M = \{v\}\cup Mand(G -\{v\}) \setminus W$. 

\begin{lemma}\label{lemma:M-subset-mand-G}
    $M \subset Mand(G)$
\end{lemma}

\begin{proof}
    First, as $v$ is simplicial, by Lemma~\ref{lemma: simplicial vertices are mandatory}, we have $v \in Mand(G)$. Moreover, by Lemma~\ref{lemma:u-mand-in-G-v-equiv-u-mand-in-G}, any vertex in $Mand(G -\{v\}) \setminus(N(v))$ is also mandatory in $G$. Finally, if $u \in (N(v) \cap M)$, then by hypothesis $u \in Mand(G -\{v\}) \setminus W$, and thus $u$ is also mandatory in $G$. Thus, $M$ only contains mandatory vertices, and $M \subset Mand(G)$. 
\end{proof}

\begin{lemma}\label{lemma:2-connect}
    $G$ is $2$-vertex connected.
\end{lemma}

\begin{proof}
    First, by minimality of $G$, $G$ is connected, otherwise the union of the meg-sets of each of its connected components would be a meg-set of $G$, and would only contain mandatory vertices.

    Suppose that $G$ is not $2$-vertex connected. Thus, there exists $u \in V(G)$ such that $G-\{u\}$ is disconnected. Let $C_1, \dots, C_p$ be the connected components of $G-u$, and let $G_1, \dots, G_p$ the subgraph induced by $V(C_1) \cup \{u\}, \dots, V(C_p) \cup \{u\}$. By minimality of $G$, since $G_1, \dots, G_p$ are chordal, $M_1 = Mand(G_1), \dots, M_p = Mand(G_p)$ are meg-sets of $G_1, \dots, G_p$.  Let $M' = (M_1 \cup \dots \cup M_p) - \{u\}$. By Lemma~\ref{lemma:cut-vertices}, $M'$ is a meg-set of $G$. It remains to prove that $M'$ only contain mandatory vertices.

    If $v_i \in V(G_i) - \{u\}$ is mandatory, then there exists $w_i$ such that any induced $2$-path $w_iv_ix$ is in a $4$-cycle. As $v_i \ne u$, any induced $2$-path $w_iv_ix$ in $G$ is an induced $2$-path in $G_i$, and thus the $4$-cycle in $G -\{v\}$ is also in $G$. which proves that $v_i$ is mandatory in $G$.

    Finally, we proved that if $G$ is not $2$-vertex connected, then either one of its $2$-connected components is a smaller counter example, or $G$ itself is not a counter example, which concludes the proof.
\end{proof}

\begin{lemma}\label{lemma:ab-monitors-e-implies-all-shortest-path-to-an-extremity}
    Let $G = (V,E)$ be a graph and $e = ab$ be an edge of $G$. Let $\{\alpha,\beta\}$ be two vertices such that $\{\alpha, \beta\}$ monitors $e$. Then, all the shortest paths from $\alpha$ to $a$ contain $e$ or all the shortest paths from $\beta$ to $a$ contain $e$.
\end{lemma}

\begin{proof}
    Since $M$ is a meg-set of $G$, there must exist two vertices $\alpha, \beta$ monitoring $e$. That is, all shortest paths from $\alpha$ to $\beta$ contain $e$. If one of $\alpha, \beta$ is $a$, then we are done. Otherwise, it must be that $\alpha$ is closer to one endpoint of $e$, say $a$ than to $b$. Observe then that, since all shortest paths from $\alpha$ to $\beta$ contain $e$ they in particular contain $a$. Hence, since $a$ is closer to $\alpha$ than $b$, a shortest path from $\alpha$ to $\beta$ can be decomposed into a shortest path from $\alpha$ to $a$, the edge $ab$ and a shortest path from $b$ to $\beta$. In particular, $\{a, \beta\}$ monitors $e$.
\end{proof}

\begin{corollary}\label{corollary:a-in-M-implies-a-monitors-incident-edges}
        Let $G= (V,E)$ be a graph, $M$ a meg-set of $G$. If $a \in M$, then for any edge $e$ incident to $a$, there exists $b \in M$ such that $\{a,b\}$ monitors $e$.
\end{corollary}

\subsection{Proof of the main result}

In order to prove Theorem~\ref{thm:mainthm}, we prove the following lemma that directly implies it. The proof of Lemma~\ref{lemma:structure-of-Mand(G)} is subdivided in four lemmas (Lemma~\ref{lemma:case1} to~\ref{lemma:case4}).

\begin{lemma}\label{lemma:structure-of-Mand(G)}
 $Mand(G) = \{v\} \cup Mand(G-\{v\})\setminus W$ and is a meg-set of $G$.
\end{lemma}

 Recall that $W \subset N(v)$ is the set of vertices that are mandatory in $G -\{v\}$ but not in $G$. We prove that $M = \{v\} \cup Mand(G - \{v\})\setminus W$ is a meg-set of $G$. Since $M$ only contains mandatory vertices by Lemma~\ref{lemma:M-subset-mand-G}, this is enough to yield Theorem~\ref{thm:mainthm}. 
 In what follows, we introduce notations that we use throughout Lemma~\ref{lemma:case1} to Lemma~\ref{lemma:case4}.
Let $e$ be an edge of $G$. The proof will be a case distinction depending on where $e$ lies in $G$. 
Finally, for a path $P = a_1a_2 \dots a_k$, we denote $|P|$ the length of $P$, which is the number of edges in $P$.

\begin{proof}
We treat cases based on the following disjunction:
\begin{itemize}
    \item either $e\in E(G-\{v\})$, and
    \begin{itemize}
        \item $e$ is monitored (in $G-\{v\}$) by two vertices of $Mand(G-\{v\})\setminus W$~(Lemma~\ref{lemma:case1});
        \item $e$ is monitored (in $G-\{v\}$) by a vertex of $Mand(G-\{v\})\setminus W$ and a vertex of $W$~(Lemma~\ref{lemma:case2}) which are either non-adjacent (first case) or adjacent (second case);
        \item $e$ is monitored (in $G-\{v\}$) by two vertices of $W$~(Lemma~\ref{lemma:case3}).
    \end{itemize}
    \item or $e\notin E(G-\{v\})$, that is $e$ incident to $v$~(Lemma~\ref{lemma:case4}).
\end{itemize}
Since $Mand(G-\{v\})$ is a meg-set of $G-\{v\}$, then every $e\in E(G-\{v\})$ is monitored by two vertices of $Mand(G-\{v\})$, hence all cases are covered.

\end{proof}

\begin{lemma}\label{lemma:case1}
    If $e\in E(G-\{v\})$, and $e$ is monitored in $G-\{v\}$ by two vertices of $Mand(G-\{v\})\setminus W$, then $e$ is monitored by two vertices of $Mand(G)$.
\end{lemma}
\begin{proof}
    If $e \in E(G -\{v\})$ such that $e$ is monitored by two vertices $\{a,b\}$ in $G -\{v\}$ that are in $Mand(G -\{v\}) \setminus W$, then by Lemma~\ref{lemma:e-monitored-by-ab-nonadjacent-to-v}, $\{a,b\}$ monitors $e \in E(G)$.
\end{proof}
\begin{lemma}\label{lemma:case2}
    If $e\in E(G-\{v\})$, and $e$ is monitored by a vertex of $Mand(G-\{v\})\setminus W$ and a vertex of $W$, then $e$ is monitored by two vertices of $Mand(G)$. 
\end{lemma}
\begin{proof}

If $e \in E(G -\{v\})$ such that $e$ is monitored by two vertices $\{a,b\}$ in $G -\{v\}$ with $a \in W$ and $b \in Mand(G -\{v\}) \setminus W$, we have the following two cases depending whether $a$ and $b$ are adjacent or not.
    \begin{itemize}
        \item We first treat the sub-case where $a$ and $b$ are non-adjacent. Note that $b$ cannot be adjacent to $v$ since $v$ is simplicial and adjacent to $a$.  If all the shortest path from $v$ to $b$ start with the edge $va$, then $e$ is on all the shortest path from $v$ to $b$ and thus it is monitored by $\{v, b\}$.
        Suppose this is not the case. Hence, there is a vertex $a' \in N(v)$ and a shortest path going from $v$ to $b$ that starts with the edge $va'$.
        Let $P = vaa_1\dots a_{k-1}b$ be a shortest path of length $k$ from $a$ to $b$ going through $a$, and $P' = va' a_1' \dots a_{\ell-1}' b$ be a shortest path from $v$ to $b$ going through $a'$ and avoiding $e$. Note that if such a path $P'$ does not exist, then $\{v, b\}$ monitors $e$. Since $P'$ is a shortest path, we have $|P'| \le |P|$. Thus, $\ell \le k$. Moreover, since $aa'a_1' \dots a_{\ell-1}'b$ is not a shortest path, from $a$ to $b$, we have $k < \ell+1$, hence $k = l$.

        Now, $a$ cannot have any neighbor $a'_i$, otherwise $aa'_i a'_{i+1}\dots a_{\ell-1}b$ would also be a shortest path from $a$ to $b$, and $a$ cannot have a neighbor $a_i$ for $i \ne 1$ as $P$ is chordless. Thus, by Lemma~\ref{lemma:induced-2-path-in-a-cycle-has-a-chord}, applied to $a_1aa'$ in the  cycle $aa_1\dots c \dots a'_1a'a$, where $c$ is the first vertex common to $P$ and $P'$ other than $v$, we must have an edge $a_1a'$. Note that such a vertex $c$ exists as $b\in P\cap P'$.

        We prove that under this conditions, either $a$ is mandatory or we will find a vertex $y$ such that $\{y,b\}$ monitors $e$.
        Let $x \in N(a)$ be a vertex. If $x \in N[v]$, since $a'$ is a neighbor of $a_1$, $a_1axa'a_1$ is a $4$-cycle, as $v$ is simplicial. Otherwise, let $y_x$ be a vertex such that $\{a, y_x\}$ monitors the edge $ax$ in $G-v$. Note that this vertex exists by Lemma~\ref{lemma:ab-monitors-e-implies-all-shortest-path-to-an-extremity}. Under this condition, $y_x$ cannot be in $N(v)$, otherwise it would be adjacent to $a$, contradicting that all shortest paths from $a$ to $y_x$ goes through the edge $ax$. If $\{y_x, b\}$ monitors $e$, then we are done as $y_x$ cannot be in $N(v)$, thus cannot be in $W$ by Corollary~\ref{corr:W}.

        Otherwise, let $P_x$ be a shortest path from $a$ to $y_x$, and denote it by  $P_x = ax x_1 \dots x_{m-1}y_x$. Let $P_y$ be a shortest path from $y_x$ to $b$ that does not contain $e$ (possibly of length $0$ if $y_x =b$). In particular $P_y$ cannot contain $a$ as all shortest paths from $a$ to $b$ contain $e$. Hence, there exists a cycle $C$ containing only vertices from $P_x, P_y$ and $P$. Note that by construction, this cycle has to contain $a, x_1$ and $a_1$ as $P_y$ cannot contain $a$. If there are no chords adjacent to $a$ in $C$, $x_1a_1$ is an edge by Lemma~\ref{lemma:induced-2-path-in-a-cycle-has-a-chord}, and $x_1xaa_1$ is a $4$-cycle. Otherwise there is a chord from $a$ to some vertex $z$ in $C$. Since both $P$ and $P_x$ are chordless and contain $\{a\}$, we necessarily have $z \in P_y$. Denote by $P_1$ the path from $z$ to to $y_x$ and by $P_2$ the path from $z$ to $b$, such that $P_y = P_1P_2$. Since $aP_1$ is not a shortest path form $a$ to $y_x$ as it does not contain $ax$, we have $|P_1| \ge m$, and since $P_2$ is not a shortest path from $a$ to $b$ as it does not contains $e$, we have $|P_2| \ge k$. But, since $P_y$ is a shortest path from $y_x$ to $b$, we need to have $|P_1| + |P_2| \le |P_x| + |P|$, \textit{i.e.} $|P_1|+ |P_2| \le k+m$, which implies $|P_1| = m$ and $|P_2| = k$. Hence, $a$ cannot be adjacent to any other vertex of $P_1$ ($P_2$ resp.) as otherwise, it would create a shortest path from $a$ to $y_x$ (to $b$ resp.), that does not contain $ax$ ($e$ resp.). We denote $c_1$ the vertex common to $P_1$ and $P_x$ that is the farthest away from $y_x$, so that $P_1 = y_x\dots c_1 P_1'$ and $P_x = y_x\dots c_1 P_x'$ for some paths $P_1', P_x'$ which must be disjoint. Conversely, we denote $c_2$ the vertex common to $P_2$ and $P$ that is the farthest away from $b$, so that $P_2 = P_2'c_2\dots b$ and $P = P''c_2 \dots b$ for some paths $P_2', P''$ which must be disjoint. Thus, by Lemma~\ref{lemma:induced-2-path-in-a-cycle-has-a-chord}, applied in the cycle obtained from $az, P_1'$ and $P_x'$; or $az, P_2'$ and $P''$, we have $x$ and $a_1$ are both adjacent to $z$, ensuring that $xaa_1zx$ is a $4$-cycle.

        Finally, we proved that either one vertex $y_x$ satisfy that $\{y_x,b\}$ monitors $e$, or all the induced $2$-paths $a_1ax$ with $x$ a neighbor of $a$ are in $4$-cycles, and $a_1$ is a support vertex for $a$, contradicting that $a$ is not mandatory and thus in $W$.

        \begin{figure}[!ht]
            \centering
            \begin{tikzpicture}[inner sep = .7mm, scale = .8]
                \tikzstyle{every node} = [black, circle, draw, fill]

                \node[] (v) at (5, 5) [label = above right : $v$]{};
                \node[] (a) at (4.3, 4) [label = above left : $a$]{};
                \node[] (a'1) at (5.7, 4) [label = right : $a'_1$]{};

                \node[] (a1) at (4.3, 3.5) [label = below right : $a_1$]{};

                \node[] (e1) at (4.3, 2.5) []{};
                \node[] (e2) at (4.3, 2) []{};

                \node[] (ak1) at (4.3, 1)[label = above right : $a_{k-1}$]{};
                \node[] (al1) at (5.7, 1)[label = above right : $a'_{\ell-1}$]{};

                \node[] (b) at (4.3, .5)[label = below : $b$]{};

                \node[] (x) at (2.9, 4) [label = above : $x$]{};

                \node[] (y) at (1, 4.2) [label = above : $y_x$]{};

                \node[] (z) at (1.5, 2.2) [label = below left : $z$]{};

                \draw[] (v) -- (a);
                \draw[] (v) -- (a'1);

                \draw[] (a) -- (x);
                \draw[dashed] (x) -- (a1);

                \draw[] (a) -- (a1);
                \draw[] (a1) -- (a'1);
                \draw[] (a) -- (a'1);
                \draw[snake it] (a'1) -- (al1);
                \draw[] (e1) -- (e2) node[midway, right, draw = none, fill =none]{$e$};
                \draw[snake it] (a1) -- (e1);
                \draw[snake it] (e2) -- (ak1);
                \draw[] (ak1) -- (b);
                \draw[] (al1) -- (b);

                \draw[snake it] (y) -- (x);
                \draw[] (z) -- (x);
                \draw[] (z) -- (a);
                \draw[] (z) -- (a1);
                \draw[snake it] (y) to[out = -90, in  = 110] (z);
                \draw[snake it] (z) to[out = -75, in  = 160] (b);

                \draw [decorate,decoration={brace,amplitude=5pt,raise=4ex}]
  (y.north) -- (a.north) node[draw = none, fill = none, midway,yshift=3em]{$|P_x| = m$};
                \draw [decorate,decoration={brace,mirror, amplitude=5pt,raise=4ex}]
  (y.south west) -- (z.south west) node[draw = none, fill = none, midway,yshift=-1em, xshift =-5em]{$|P_1|\geq m$};
                \draw [decorate,decoration={brace,mirror, amplitude=5pt,raise=4ex}]
  (z.south west) -- (b.south west) node[draw = none, fill = none, midway,yshift=-3em, xshift =-2em]{$|P_2|\geq k$};
            \end{tikzpicture}
    \caption{Representation of the proof of Lemma~\ref{lemma:case2}, that is $e \in E( G -\{v\})$, $a \in W$, $b \notin W$, where $a$ and $b$ are not adjacent.}
            
\end{figure}

        \item We can now assume that $a$ and $b$ are adjacent, that is, $e = ab$. If $\{b,v\}$ monitors $e$, then we are done. Otherwise, if $b$ is a support for $a$, then $a$ is mandatory, which contradicts $a \in W$. Hence, there exists $a_1 \in N(a)$ such that $baa_1$ is an induced two path that is not is a $4$-cycle. If $a_1$ is mandatory, as $baa_1$ is the only $2$-path between $b$ and $a_1$, then $\{b,a_1\}$ would monitor $ab$ and we are done. Otherwise, by Lemma~\ref{lemma:ab-monitors-e-implies-all-shortest-path-to-an-extremity} there exist $a_k$ such that $\{a, a_k\}$ monitors $aa_1$ in $G -\{v\}$. In particular $a_k \notin N(v)$, otherwise, since $v$ is simplicial and adjacent to $a$, $a_k$ would be adjacent to $a$. Therefore, by Lemma~\ref{lemma:u-mand-in-G-v-equiv-u-mand-in-G}, $a_k\in M$. We denote by $P = aa_1a_2\dots a_k$ a shortest path from $a$ to $a_k$. If all the shortest path from $b$ to $a_k$ goes through $e$, then $\{b, a_k\}$ monitors $e$ and we are done. Otherwise there exist a shortest path $P' = b_\ell b_{\ell-1}  \dots b_1 b$ of length $\ell$ from $b$ to $b_\ell = a_k$ that does not go through $e$. Since it is a shortest path, we have $\ell \le k+1$. Moreover since $abb_1 \dots b_{\ell-1} a_k $ is not a shortest path, otherwise $aa_1$ would not be monitored by $\{a,a_k\}$, we have $\ell +1 > k$. All together, we have $\ell \in \{k, k+1\}$. We also denote $c$ the vertex common to $P$ and $P'$ which is farthest away from $a_k$, so that $P = aP_a c\dots a_k$ and $P' = a_k\dots cP_bb$ where $P_a$ and $P_b$ are disjoint.

        \medskip

        \underline{If $\ell = k$:} The vertex $a$ cannot have any neighbor $b_i$ in $P'$ other than $b$, otherwise $ab_i\dots b_k$ would be a shortest path from $a$ to $a_k$ that does not contain the edge $aa_1$. Hence, by Lemma~\ref{lemma:induced-2-path-in-a-cycle-has-a-chord} applied to $a_1ab$ in the cycle $aP_acP_bba$, $a_1$ and $b$ are adjacent, contradicting the fact that $baa_1$ is an induced $2$-path. A representation of this configuration is shown in the left graph of Figure~\ref{fig:aWbnotinW}.

\begin{figure}
    \centering
    \begin{tikzpicture}[inner sep = .7mm, scale =1]
        \tikzstyle{every node} = [black, circle, draw, fill]
        \node[] (v) at (0, 0)[label = above:$v$]{};
        \node[] (a) at (.5, -.7)[label = above right:$a$]{};
        \node[] (b) at (1.8, -.9)[label = above:$b$]{};
        \node[] (a1) at (0, -1.5)[label = left:$a_1$]{};
        \node[] (b1) at (1.2, -1.5)[label = right:$b_1$]{};
        \node[] (ak) at (-1, -3)[label = left : $a_k$]{};
        \node[] (bi) at (.8, -1.75)[label = below right : $b_i$]{};
        \draw[] (v)--(a);
        \draw[] (a)--(b);
        \draw[dashed] (a1) -- (b);
        \draw[] (a)--(a1);
        \draw[] (b)--(b1);
        \draw[snake it] (a1)--(ak);
        \draw[snake it] (b1)--(ak);
        \draw[red] (a) -- (bi);

        \draw [decorate,decoration={brace, mirror, amplitude=5pt,raise=4ex}]
  (a.west) -- (ak.north west) node[draw = none, fill = none, midway,xshift=-4em, yshift = 2em]{$|P| = k$};
  \draw [decorate,decoration={brace, amplitude=5pt,raise=4ex}]
  (b.east) -- (ak.north east) node[draw = none, fill = none, midway,xshift=4em, yshift = -2em]{$|P'| = l = k$};
    \end{tikzpicture}
    \begin{tikzpicture}[inner sep = .7mm, scale =1]
        \tikzstyle{every node} = [black, circle, draw, fill]
        \node[] (v) at (0, 0)[label = above:$v$]{};
        \node[] (a) at (.5, -.7)[label = above right:$a$]{};
        \node[] (b) at (1.8, -.9)[label = above:$b$]{};
        \node[] (a1) at (0, -1.5)[label = left:$a_1$]{};
        \node[] (b1) at (1.2, -1.5)[label = right:$b_1$]{};
        \node[] (ak) at (-1, -3)[label = left : $a_k$]{};
        \node[] (bi) at (.8, -1.75)[label = below right : $b_i$]{};
        \draw[] (v)--(a);
        \draw[] (a) -- (b1);
        \draw[] (a)--(b);
        \draw[] (a)--(a1);
        \draw[dashed] (a1) -- (b1);
        \draw[] (b)--(b1);
        \draw[snake it] (a1)--(ak);
        \draw[snake it] (b1)--(ak);
        \draw[red] (a) -- (bi);

        \draw [decorate,decoration={brace, mirror, amplitude=5pt,raise=4ex}]
  (a.west) -- (ak.north west) node[draw = none, fill = none, midway,xshift=-4em, yshift = 2em]{$|P| = k$};
  \draw [decorate,decoration={brace, amplitude=5pt,raise=4ex}]
  (b.east) -- (ak.north east) node[draw = none, fill = none, midway,xshift=4em, yshift = -2em]{$|P'| = k+1$};
    \end{tikzpicture}
    \caption{Representation of the proof of Lemma~\ref{lemma:case2}, that is $e\in E(G -\{v\})$, where $a\in W$, $b\notin W$, $a$ and $b$ are adjacent.}
    \label{fig:aWbnotinW}
\end{figure}
        \medskip

        \underline{If $\ell = k+1$:} The same argument as in the case $\ell = k$ works, except that now, the chord $ab_1$ can exist. Note that this is the only chord that can be attached to $a$ without breaking the minimality of $P$, the minimality of $P'$, or the fact that $\{a,a_k\}$ monitors $aa_1$. If $ab_1$ does not exist, then for the same reason as in the previous case, we would have the edge $a_1b$ contradicting that $a_1ab$ is an induced $2$-path. Now if $ab_1$ exists, consider the cycle $aP_acP_ba$. In this cycle, $a$ cannot have any other chord adjacent to it, otherwise, it would contradict the fact that $P$ or $P'$ is a shortest path. Hence, by Lemma~\ref{lemma:induced-2-path-in-a-cycle-has-a-chord}, there is an edge $a_1b_1$. Hence, $baa_1b_1b$ is a $4$-cycle, contradicting the fact that $baa_1$ is in no $4$-cycle. A representation of this configuration is shown in the right graph of Figure~\ref{fig:aWbnotinW}.

        All together, $b$ is a support for $a$ in $G$, contradicting that $a$ is not mandatory in $G$.

    \end{itemize}

\end{proof}
\begin{lemma}\label{lemma:case3}
    If $e\in E(G-\{v\})$, and $e$ is monitored by two vertices of $W$, then $e$ is monitored by two vertices of $Mand(G)$. 
\end{lemma}
\begin{proof}
    Suppose that $e \in E(G -\{v\})$ such that $e$ is monitored by two vertices $\{a,b\}$ in $G -\{v\}$ with $a,b \in W$. Intuitively, we will have the same argument as in the case where $b \notin W$ and $a$ and $b$ are adjacent, but on the two extremities of $e$. 

    \smallskip
    
    Since $W \subset N(v)$ and $v$ is simplicial, necessarily, $e = ab$ at it is the only shortest path between $a$ and $b$. Since $a$ and $b$ are in $W$, they are not mandatory. Thus, $a$ is not a support for $b$ neither is $b$ for $a$. Let $a_1$ be a neighbor of $a$ such that $baa_1$ is an induced $2$-path that is not in a $4$-cycle. Similarly, let $b_1$ be a neighbor of $b$ such that $abb_1$ is an induced $2$-path that is not in a $4$-cycle. Note that this construction implies that $a_1b, ab_1$, and $a_1b_1$ are non-edges. Let $a_n, b_m$ be two vertices such that $a_1a$ is monitored by $\{a_n,a\}$ and $b_1b$ is monitored by $\{b_m,b\}$  in $G - \{v\}$. Let $P_a = aa_1a_2 \dots a_n$ be a shortest path from $a$ to $a_n$ and $P_b = bb_1 \dots b_m$ be a shortest path from $b$ to $b_m$. Note that $a_n$ and $b_m$ cannot be in $N(v)$, otherwise, they would be adjacent to $a$ and $b$, and thus they could not monitor $aa_1$ and $bb_1$. Hence, by Lemma~\ref{lemma:u-mand-in-G-v-equiv-u-mand-in-G}, $a_n$ and $b_m$ are mandatory in $G$.

    We first prove that $P_a$ and $P_b$ are disjoint. Suppose they are not and let $c = a_i = b_j$ be the closest vertex from $a$ and $b$ in which they intersect. Necessarily, we have $i \ge 2$ and $j \ge 2$ by construction of $a_1$ and $b_1$. Up to exchange $P_a$ and $P_b$, suppose $i \le j$. By application of Lemma~\ref{lemma:induced-2-path-in-a-cycle-has-a-chord}, since $ab_1$ is a non-edge, there is a chord incident to $b$ in the cycle $aa_1\dots a_i b_{j-1} \dots b_1 b a$. Since $P_b$ is a shortest path, this chord must be from $b$ to some vertex $a_\ell$ for $2 \le \ell \le i$. Consider the path $b a_\ell a_{\ell +1} \dots a_i b_{j+1} \dots b_m$. As $ i \le j$, it is shorter that $P_b$ and does not contain $bb_1$ contradicting that $\{b, b_m\}$ monitors $bb_1$. Thus $P_a$ and $P_b$ are disjoint.

    We prove that $\{a_nb_m\}$ monitors $e$. Suppose by contradiction that this is not the case. Hence, there is a shortest path $P_c = a_n c_1, \dots c_k b_m$ shorter or equal to $a_na_{n-1} \dots a_1abb_1 \dots b_{m-1}b_m$. Hence $k+1 \le (n-1) + (m-1) + 3$, i.e. $k \le n+m$. 
    In this path, we denote $P_a'$ (resp. $P_b'$) the longest prefix (resp. suffix) of $a_n c_1, \dots c_k b_m$ ending (resp. starting) with a vertex $c_a$ of $P_a$ (resp. $c_b$ of $P_b$), so that we can define $P_c'$ with $a_n c_1, \dots c_k b_m = P_a'c_aP_c'c_bP_b'$. Note that $C= aa_1\dots c_aP_c'c_b\dots b_1ba$ is a cycle of length at least $4$.
    
    Since $G$ is chordal and $a_1b$ is a non-edge, by Lemma~\ref{lemma:induced-2-path-in-a-cycle-has-a-chord} any cycle containing $a_1ab$ contains a chord incident to $a$. Similarly, any cycle containing $abb_1$ contains a chord incident to $b$. In particular, this is true for the cycle $C$. Since $aa_1\dots a_n$ is a shortest path, their is no chord $aa_i$. Moreover, since $bb_1\dots b_m$ is a shortest path, and $ab_1$ is a non-edge, for $1 \le i \le m$, there is no chord $ab_i$, as otherwise there would be a shortest path $bab_i b_{i+1}\dots b_m$ that does not contain $bb_1$. Similarly, there are no chords $ba_i$ nor $bb_j$. Hence, both $a$ and $b$ have a chord to some vertices $c_i$ and $c_j$ of $P_c'$.

    \smallskip

    We prove now that there exists $1 \le i \le k$ such that $c_i$ is adjacent to $a$ and to $b$. Let $1 \le i,j \le k$ such that $a$ is adjacent to $c_i$, $b$ is adjacent to $c_j$, $c_i, c_j \in P_c'$ and $|j-i|$ is minimized. Suppose by contradiction $i \neq j$, say $i<j$ without loss of generality, the case $i > j$ being similar. The cycle $ac_ic_{i+1}\dots c_{j-1}c_jba$ has length at least $4$, and contains no chord, as, for $i < \ell,\ell' < j$ there are no chord $c_\ell c_{\ell'}$ by minimality of the path $c_1\dots c_m$, and no chord $ac_\ell$ nor $b c_\ell$ by minimality of $|j-i|$, which contradicts that $G$ is chordal. Thus, there exists an integer $i$ such that $a$ and $b$ are adjacent to $c_i$. Since $\{a, a_n\}$ monitors $aa_1$, we have $i +1 \ge n +1$, and since $\{b, b_m\}$ monitors $bb_1$, we have $k-i +2 \ge m +1$. We now prove that these two inequalities are strict. We prove it for the first one, the similar argument works for the second. Suppose $i = n$, and consider the cycle $aa_1\dots c_a\dots c_ia$. By minimality of the paths, $a$ cannot be adjacent to any other $c_j$ with $1 \le j \le i$, nor to any $a_p$ for $1 \le p \le n$. Thus, by Lemma~\ref{lemma:induced-2-path-in-a-cycle-has-a-chord}, $a_1c_i$ is an edge. Thus $a_1abc_ia_1$ is a $4$-cycle which is a contradiction with the definition of $a_1$. Thus, $i \ge n +1$. We prove similarly that $k-i+2 \ge m+2$
    
    Combining the two inequalities, we have $k+2 \ge m+n +3$, i.e. $k \ge m+n+1$, which contradicts the fact that $a_nc_1\dotsc_kb_m$ is a shortest path from $a_n$ to $b_m$, as it is longer than $a_n a_{n-1}\dots a_1abb_1 \dots b_{m-1}b_m$ and thus, $\{a_n, b_m \}$ monitors $e$. A representation is shown in Figure~\ref{fig:abW}.
    \begin{figure}
        \centering
        \begin{tikzpicture}[inner sep  = .7mm, scale = 1]
                    \tikzstyle{every node} = [black, circle, draw, fill]
                \node[] (v) at (0, 0)[label = above : $v$]{};
                \node[] (a) at (-.5, -.7)[label = above left: $a$]{};
                \node[] (b) at (.5, -.7)[label = above right: $b$]{};
                \node[] (a1) at (-1.2, -1.2)[label = left: $a_1$]{};
                \node[] (b1) at (1.2, -1.2)[label = right: $b_1$]{};
                \node[] (an) at (-2.5, -3)[label = below left: $a_n$]{};
                \node[] (bm) at (2.5, -3)[label = below right: $b_m$]{};
                \node[] (ci) at (-1, -3)[label = above left : $c_i$]{};

                \draw[] (a)--(b);
                \draw[] (a)--(v);
                \draw[] (b)--(v);
                \draw[] (a)--(a1);
                \draw[] (b)--(b1);
                \draw[red] (a)--(b1);
                \draw[red] (b)--(a1);
                \draw[red] (a1)--(b1);

                \draw[snake it] (a1)--(an);
                \draw[snake it] (b1)--(bm);
                \draw[snake it] (an)--(bm);
                \draw[] (a)-- (ci);
                \draw[] (b)--(ci);

                \draw [decorate,decoration={brace, mirror, amplitude=5pt,raise=4ex}]
  (a.west) -- (an.west) node[draw = none, fill = none, midway,xshift=-4em, yshift = 2em]{$|P_a| = n$};
                \draw [decorate,decoration={brace, amplitude=5pt,raise=4ex}]
  (b.east) -- (bm.east) node[draw = none, fill = none, midway,xshift=4em, yshift = 2em]{$|P_b| = m$};
                \draw [decorate,decoration={brace, mirror, amplitude=5pt,raise=4.5ex}]
  (an.south) -- (ci.south) node[draw = none, fill = none, midway, yshift = -3em]{$i > n$};
                \draw [decorate,decoration={brace, mirror, amplitude=5pt,raise=4.5ex}]
  (ci.south) -- (bm.south) node[draw = none, fill = none, midway, yshift = -3em]{$k-i > m$};
                \draw [decorate,decoration={brace, mirror, amplitude=5pt,raise=2ex}]
  (an.north) -- (bm.north) node[draw = none, fill = none, midway, yshift = -2em]{$|P| = k$};
        \end{tikzpicture}
        \caption{A representation of the proof of Lemma~\ref{lemma:case3}, that is $e\in E(G -\{v\})$, where $a\in W$ and $b\in W$.}
        \label{fig:abW}
    \end{figure} 
\end{proof} 
\begin{lemma}\label{lemma:case4}
    If $e\notin E(G-\{v\})$, then $e$ is monitored by two vertices of $Mand(G)$.
\end{lemma}
\begin{proof} 
Assume $e = va \notin E(G -\{v\})$. Recall that $v$ is simplicial, thus mandatory by Lemma~\ref{lemma: simplicial vertices are mandatory}. If all the vertices $a_1 \in N(a) \setminus N(v)$ have a neighbor $w$ in $N(v) \setminus \{a\}$, by definition, $v$ is a support for $a$, end thus, $a$ is mandatory, and $\{a,v\}$ monitors $e$.
     
     Otherwise, there exists $a_1$ in $N(a) \setminus N(v)$ that has no neighbor in $N(v)\setminus \{a\}$. Since $a_1$ is not incident to $v$, using the previous points of the proof, $aa_1$ is monitored by two vertices $\{b,a_n\} \in M$. Using Lemma~\ref{lemma:ab-monitors-e-implies-all-shortest-path-to-an-extremity}, we can suppose that all the shortest path from $a_n$ to $a$ contains $aa_1$. We denote one of these paths by $aa_1\dots,a_n$
     We prove that $\{a_n,v\}$ monitors $e$. Suppose that they don't. There exists another shortest path $va'a'_1\dots a'_{m-1}a_n$ that monitors $e$. By hypothesis, we have $a' \ne a$, otherwise, this path would also go through $e$. Since $v$ is simplicial, $aa'$ is an edge, and since $a_1$ has been chosen such that $a_1$ has no neighbors in $N(v)$, $a'a_1$ is not an edge.
     
     As this path is a shortest path, we have $m \le n$, otherwise it would be longer that $vaa_1\dots a_n$. Moreover, $a_na'_{m-1}\dots a'a$ cannot be a shortest path from $a$ to $a_n$, by hypothesis on $\{a,a_n\}$. Thus $m+1 >n$. This proves that $m=n$. Let us denote $c$ the smallest value $1\leq k\leq n$, such that $a_k = a'_k$ (that is, the first common vertex between $a'_1\dots a_{n-1}'a_n$ and $aa_1\dots,a_n$). Consider now the cycle $C=a_ka'_{k-1}\dots a'_1a'aa_1\dots a_k$. Since $a_1 a'$ is a non-edge, $C$ has length at least $4$, and there is a chord incident to $a$ in that cycle by Lemma~\ref{lemma:induced-2-path-in-a-cycle-has-a-chord}. This is a contradiction as, for $1 \le i \le n$, if this chord goes to a vertex $a_i$, in would contradict the minimality of $aa_1\dots a_n$, and if it goes to a vertex $a'_i$, the path $aa'_ia'_{i+1}\dots a_n$ would be shorter or equal to $aa_1\dots a_n$ without containing $aa_1$. A representation of this configuration is shown in Figure~\ref{fig:av}. This proves that $\{a_n,v\}$ monitors $e$. 
\end{proof}

             \begin{figure}
            \centering
            \begin{tikzpicture}[inner sep = .7mm, scale = .85]
                \tikzstyle{every node} = [black, circle, draw, fill]

                \node[] (v) at (5, 4.5) [label = above right : $v$]{};
                \node[] (a) at (4.3, 4) [label = above left : $a$]{};
                \node[] (a'1) at (5.7, 4) [label = right : $a'$]{};

                \node[] (a1) at (3.3, 4) [label = below : $a_1$]{};

                \node[] (b) at (7.3, 4.5)[label = above right : $b$]{};

                \node[] (an) at (1, 4.2) [label = above : $a_k$]{};
                \node[] (atn) at (0, 4.2) [label = above : $a_n$]{};
                \node[] (a'2) at (5, 3.3)[label = below right : $a'_i$]{};

                \draw[] (v) -- (a);
                \draw[] (v) -- (a'1);
                \draw[snake it] (an) to (atn);

                \draw[] (a) -- (a1);
                \draw[] (a1) -- (a'1);
                \draw[] (a) -- (a'1);

                \draw[snake it] (a) to[out = 90, in  = 110] (b);
                \draw[snake it] (an) -- (a1);
                \draw[snake it] (an) to[out = -40, in  = -130] (a'1);
                \draw[red] (a1) to[out = -30, in  = -150] (a'1);

                \draw[red] (a) -- (a'2);

            \end{tikzpicture}
            \caption{A representation of the proof of Lemma~\ref{lemma:case4}, that is  $e\notin E(G -\{v\})$, where $a\in W$ and $b=v$.}
            \label{fig:av}
        \end{figure}
Finally, we have proved that any edge $e \in E(G)$ is monitored by $M$, hence $M$ is a meg-set. Since all the meg-sets of a graph has to contain all its mandatory vertices, it proves that $mand(G) \subset M$. Overall, it proves that $M = Mand(G)$ is a meg-set of $G$.

\section{An efficient algorithm to compute $Mand(G)$}

In this section, we provide an algorithm to computes $Mand(G)$ whose complexity time is $O\big (|V|(|V| + \Delta^2) \big)$ for any graph $G = (V,E)$. Using the structure provided in Lemma~\ref{lemma:structure-of-Mand(G)}, this algorithm returns a minimum meg-set of $G$ if $G$ is chordal.

\begin{theorem}\label{th:algo}
    Let $G$ be a graph. $Mand(G)$ can be computed in time  $O\big(|V|(|V| + \Delta^2) \big )$, where $\Delta$ is the maximum degree of $G$. In particular, if $G$ is chordal, a minimum meg-set of $G$ can be computed in time $O\big(|V|(|V| + \Delta^2) \big )$.
\end{theorem}

\begin{proof}
Let $G = (V,E)$ be a graph. We suppose that the graph is connected, and given by its adjacency matrix $A$.  We denote the vertices of $G$ by $v_1, \dots v_n$ 
    We propose the following algorithm:

\begin{itemize}
    \item For $1\le i \le n$ and each vertex $v_i$ we compute its set of neighbors $N(v_i)$, which can be done in $O(|V|)$ per vertex. 
    \item Then we create a table $N^i_{2}$ such that $N_2^i[j]$ contains the number of paths of length 2 between $i$ and $j$. This can be done in time $O(\Delta^2)$ per vertex by initializing this matrix with zeros and then adding $1$ in $N_2^i[j]$ for each integer $k \in N(v_j)$ such that $A[j,k] = 1$.
    \item We now compute $M$ as a table of size $n$ by induction. For each vertex, we will also keep in memory the list of its current support vertices. I.e. we start with $S(v_i) = N(v_i)$ for all $1 \le i \le n$.
    \item For $i = 1$ to $n$:
    \begin{itemize}
        \item If $|S(v_i)| \ge 1$, set $M[i]$ to $1$.
        \item For each vertex $v_j \in N(v)$, for each vertex $v_k \in S(v_j)$, if $N_2^i[k] = 1$ and $A[i,k]=0$, remove $v_k$ for $S(v_j)$. Moreover, if $S(v_j) = \emptyset$, set $M[j]$ to 0.  This step can be done in $O(\Delta^2)$. Note that we count the number of path of length $2$ in $G$ and not in $G_i$, thus, it is possible that during this step $M$ is not a meg-set of $G_i$. 
    \end{itemize}
    \item Finally, returns the set of vertices with nonzero value in $M$.
\end{itemize}

We now prove that the set of vertices in $M$ it the set of mandatory vertices of $G$. Let $v_i \in Mand(G)$. During step~$i$, $M[i]$ is turned to $1$. Now, if $v_j$ is a support for $v_i$, for $1 \le k \le n$ during step $k$, if $v_k \in N(v_i)$, either $v_k \in N(v_j)$ or there are at least two paths of length~$2$ from $v_j$ to $v_k$ as $v_j$ is a support vertex for $v_i$. Thus, $v_j$ will not be removed from $S(v_i)$, and thus $v_i$ will never be removed from $M$.

Reciprocally, suppose that $v_j \notin Mand(G)$. Let $v_k \in N(v_j)$ be a vertex. As $v_k$ is not a support for $v_j$, there exist $v_i \in N(v_j) \setminus N(v_k)$ such that $v_jv_iv_k$ is an induced $2$-path that is not in a $4$-cycle. In particular, $v_jv_iv_k$ is the only path of length~$2$ between $v_j$ and $v_k$. During step $i$, when $v_j$ will be considered in $N(v_i)$, if $v_k$ has not been removed yet form $S(v_j)$, we will have $N_2^i[k] = 1$ and $A[i,k] =0$ by hypothesis. Thus, $v_k$ will be removed from $S(v_j)$. Since this is true for all the vertices in $N(v_j)$, all the vertices of $S(v_j)$ will be removed at some points, and thus $M[j]$ will be set to $0$ (or will not be set to $1$ if $v_j$ comes after all its neighbors).

Finally, at the end of the algorithm, all the vertices with nonzero values in $M$ are mandatory, and the algorithm runs in $O\big(|V|(|V| + \Delta^2) \big )$
\end{proof}

\section{Conclusion}
The main contribution of this work is its answer to the open conjecture of~\cite{foucaud2024algorithms} about chordal graphs. In addition, Theorem~\ref{th:algo} provides a dedicated algorithm to compute the set of mandatory vertices of any graph $G$, and thus the minimum meg-set of a chordal graph. From our main result, as well as the structural insights gained through its proof, several research directions naturally emerge:
\begin{itemize}
    \item A first, natural way to go beyond Theorem~\ref{thm:mainthm} is to try and extend it to other classes of graphs. In particular, with the notable exception of cycles, all current algorithms to solve \textsc{MEG-set} in polynomial time are restricted to meg-minimal graph classes. This motivates a systematic study of the class of meg-minimal graphs. We proved that this class contains chordal graphs, and it is already known that it contains the class of cographs. It would be of interest to search for other structural characterizations of it. Another underlying question is to find other methods to compute minimum meg-sets of non meg-minimal graphs in polynomial time.

    \item An insight that can be gathered from our result, as well as the meg-minimality of cographs, is that the $meg$ parameter seems to interact with density. This phenomenon was already observed in~\cite{foucaud:hal-04494089} for graphs with high girth. While some approaches already tackled graph sparsity parameters (see e.g.~\cite{foucaud2024algorithms}), a parameterized study of \textsc{MEG-set} by density parameters, like the independence number, is both relevant regarding our current knowledge, and well-fitting into some more recent framework (see e.g.~\cite{fomin2025pathcoverhamiltonicityindependence}). Moreover, should this kind of parameterization exist, it would open a lot of possibilities to design algorithms to compute meg-sets in graphs of bounded tree-independence number. This is further supported by the fact that chordal graphs are the graphs of tree-independence number $1$. 
\end{itemize}



%
%

\bibliographystyle{splncs04}
\bibliography{mybibliography}

\end{document}